\documentclass[12pt]{article}
\usepackage[english]{babel}
\usepackage{amssymb}
\usepackage{amsmath}
\usepackage{amsfonts}
\usepackage{amsbsy}
\usepackage{indentfirst}
\usepackage{graphicx}
\usepackage{color}
\newcommand{\fimdem}{$\hfill \blacksquare$}
\newcommand{\nd}{\noindent}
\newtheorem{teor}{Theorem}[section]
\newtheorem{pro}{Proposition}[section]

\newtheorem{corol}{Corollary}[section]
\newtheorem{lema}{Lemma}[section]
\newtheorem{Exa}{Example}[section]

\newenvironment{proof}[1][Proof]{\noindent\textbf{#1.} }{\ \rule{0.5em}{0.5em}}
\begin{document}
\newcommand{\bfA}{\mathbf{A}}
\newcommand{\C}{\mathbb{C}}
\newcommand{\R}{\mathbb{R}}
\newcommand{\N}{\mathbb{N}}
\newcommand{\bfgr}{\boldsymbol{\nabla}}
\newcommand{\bfal}{\boldsymbol{\alpha}}
\newcommand{\bfpi}{\boldsymbol{\pi}}
\newcommand{\bfta}{\boldsymbol{\tau}}
\newcommand{\bfr}{\mathbf{r}}
\newcommand{\bfq}{\mathbf{q}}
\newcommand{\e}{\mathrm{e}}
\def\d{{\rm d}}
\def\tr{{\rm tr}}
\def\D{\rho}
\def\Tr{{\rm Tr}}
\def\hS{{\hat S}}
\def\al{\alpha}
\def\be{\beta}
\def\la{\lambda}
\def\ga{\gamma}
\font\titlefont=cmbx10 scaled\magstep1
\title{
Inequalities in
R\'enyi's
quantum thermodynamics}
\author{
N. Bebiano\footnote{ CMUC, University of Coimbra, Department of
Mathematics, P 3001-454 Coimbra, Portugal (bebiano@mat.uc.pt)},
 J.~da Provid\^encia\footnote{CFisUC, University of Coimbra, Department of
Physics, P 3004-516 Coimbra, Portugal (providencia@teor.fis.uc.pt)}~
and J.P. da
Provid\^encia\footnote{Department of Physics, Univ. of Beira
Interior, P-6201-001 Covilh\~a, Portugal
(joaodaprovidencia@daad-alumni.de)}
}
\maketitle
\begin{abstract}
A theory of thermodynamics has been recently formulated
and derived on the basis of R\'enyi entropy and its relative versions.
In this framework, we define the concepts of partition function, internal energy and
free energy, and fundamental quantum thermodynamical
inequalities are deduced. In the context of R\'enyi's thermodynamics,
the variational Helmholtz principle
is stated and the condition of equilibrium is analyzed.
The R\'enyi maximum entropy principle is formulated and the equality case is discussed.
The obtained results reduce to the von Neumann ones when the R\'enyi entropic parameter $\alpha$ approaches 1.
The Heisenberg and Schr\"odinger uncertainty principles on the measurements
of quantum observables are revisited.
The presentation is self-contained and the proofs only use
standard matrix analysis techniques.
\end{abstract}

{\bf Keywords:} R\'enyi entropy, R\'enyi relative entropy,
partition function, Helmholtz free energy,
uncertainty principles,

\section{Introduction}

A complete theory of thermodynamics has been recently formulated
and derived on the basis of the R\'enyi entropy and its relative version \cite{misra},
which are crucial, for instance, in defining the laws of
quantum thermodynamics at microscopic level. This fact is a relevant manifestation
of the incidence of information theory concepts in thermodynamics when
extended to the quantum context.

We consider a quantum system possessing a given Hamiltonian $H,$
defined in a complex Hilbert space
with finite dimension, and being described by an arbitrary
{\it density matrix} $\rho$, i.e., a positive definite matrix with trace 1.
In statistical physics, isolated systems
are described by {\it microcanonical ensembles} and systems in equilibrium with a heat
bath are described by {\it canonical ensembles}.
The {canonical ensemble} is not adequate for the statistical description of
systems with a small number of particles compared with Avogadro's number,
such as a DNA molecule, while the microcanonical ensemble is hard to handle.
This fact explains the interest on statistical descriptions based on
different definitions of entropy from the von Neumann entropy, such as the Tsallis
 or the R\'enyi's entropies.

According to classical thermodynamics, the entropy of a thermally
isolated system is maximal for the equilibrium state ({\it maximum entropy principle}).
The Helmholtz free energy of a system in thermal contact with its environment,
or with a heat bath characterized by a temperature $T$, is minimal for the equilibrium
state ({\it minimum free energy principle}).

This paper is organized as follows. In Section \ref{S2} we define the R\'enyi internal energy and the R\'enyi entropy of a physical
system in terms of the density matrix $\rho$, and, in accordance with the principles of thermodynamics,
we determine the state of equilibrium of the system by minimizing, at constant
temperature, the Helmholtz free energy. In section \ref{S3}, the close relation between
the R\'enyi relative entropy and the Helmholtz free energy is discussed.
Since the state described by the density matrix $\rho$ is completely arbitrary, it is not
characterized by a well defined temperature. In Section \ref{S4},
we investigate the relation
between the partition function and the internal energy,
for arbitrary temperature.
In Section \ref{S5}, the R\'enyi maximum entropy principle is formulated.
In Section \ref{S6}, the uncertainty principle on the measurements of
quantum observables is revisited. In Section \ref{S7}, the obtained results
are discussed and some open problems are formulated.

\section{R\'enyi's entropies}\label{S2}
\subsection{General properties}
Let $M_n$ be the matrix algebra of $n\times n$ matrices with complex entries and
$H_n$ the vector space of Hermitian matrices, named in physics as {\it observables}. By $H_{n,+}$ we denote
the cone  of Hermitian positive definite matrices and $H_{n,+,1}$ consists of
positive Hermitian matrices with unit trace, called the {\it state space}.
This set coincides with the class of {\it density matrices}
acting on an $n\times n$ quantum system,
and we use the terms state and density matrix synonymously.
Matrices in $H_{n,+}$ with rank one describe
{\it pure} states and those with rank greater than one represent {\it mixed} states.

Throughout we use the conventions $0\log0=0,~\log0=-\infty$ and $\log\infty=\infty.$
For a density matrix $\rho$ with eigenvalues $\rho_1\geq\ldots\geq\rho_n,$
the $\alpha$-{\it R\'enyi  entropy} \cite{misra} is defined as
\begin{equation}\label{Salpha}
S_\alpha(\rho)=:{\log\Tr\rho^\alpha\over 1-\alpha}={\log\sum_{i=1}^n\rho_i^\alpha\over1-\alpha},\quad\alpha\in(0,1)\cup(1,\infty).
\end{equation}
If $\alpha>1$, then $\Tr\rho^\alpha<1$ and so $\log\Tr\rho^\alpha<0.$  If $\alpha<1$,
we have $\Tr\rho^\alpha>1$ and consequently $\log\Tr\rho^\alpha>0.$
Hence, $S_\alpha(\rho)\geq0$ for any $\rho$, and equality holds if and only if $\rho$ is a pure state.
For $\rho_1=\ldots =\rho_n=1/n$,
we obtain $S_\alpha(\rho)=\log n$, which is the maximum possible value of $S_\alpha(\rho)$. Therefore,
$$0\leq S_\alpha(\rho)\leq\log n.$$

To avoid dividing by zero in (\ref{Salpha}), we consider $\alpha\neq1,$ but l'H\^opital rule
shows that the $\alpha$-{R\'enyi entropy} approaches the Shannon entropy $S_1$ \cite{Shannon '49} as $\alpha$
approaches 1:
$$S_1(\rho)=\lim_{\alpha\rightarrow1}S_\alpha(\rho)=-\Tr\rho\log\rho.$$

The special cases $\alpha=0$ and $\alpha=\infty$ may be defined by taking the limit.
In physics, many uses of R\'enyi entropy involve the limiting cases $S_0(\rho)=\lim_{\alpha\rightarrow0}S_\alpha(\rho)$
and  $S_\infty(\rho)=\lim_{\alpha\rightarrow\infty}S_\alpha(\rho)$, known as
``max-entropy" and ``min-entropy", as $S_\alpha(\rho)$ is a monotonically decreasing function of $\alpha:$
$$S_\alpha(\rho)\leq S_{\alpha'}(\rho)~~\text{for}~~\alpha< \alpha'.$$
Min-entropy is the smallest entropy measure in the class of R\'enyi entropies and
it is the strongest measure of information content of  a discrete quantum variable.
It is never larger than the Shannon entropy $S_1$.

A function $g:H_n\rightarrow\R$, is {\it  concave}
if, for $A_1,A_2\in H_n$, $0\leq p\leq1$, the following holds,
$$g(pA_1+(1-p)A_2)\geq p g(A_1)+(1-p)g(A_2).$$
\begin{teor}
R\'enyi's entropy map $S_\alpha: H_{n,+,1}\rightarrow\R$ for $0<\alpha<1$ is concave.
\end{teor}
\begin{proof}
This is a simple consequence of the concavity of both  $x^\alpha$, for $\alpha<1$,
and $\log x$.
\end{proof}

For $\alpha>1$, $x^\alpha$ is convex, so $S_\alpha(\rho)$ is neither purely convex nor concave.

The $\alpha$-{\it R\'enyi relative entropy} ($\alpha$-RRE) \cite{renyi} between two quantum states
$\rho\in H_{n,+,1}$ and $\sigma\in H_{n,+}$ is defined by
$$\mathcal{D}_\alpha(\rho\|\sigma)={\log\Tr(\rho^\alpha\sigma^{1-\alpha})\over\alpha-1},~~
\alpha\in(0,1)\cup(1,\infty) .$$
The special cases $\alpha=1$ and $\alpha= \infty$ are defined taking the limit.

The $\alpha$-RRE satisfies
$$D_\alpha(U^*\rho U\|U^*\sigma U)=D_\alpha(\rho\|\sigma)$$
for all unitary matrices $U$. If $\rho$ and $\sigma$ commute they are simultaneously diagonalizable and so
$$D_\alpha(\rho\|\sigma)={\sum_{i=1}^n\rho_i^\alpha\sigma_i^{1-\alpha}\over\alpha-1},$$
where $\rho_i$ and $\sigma_i$ are respectively the eigenvalues (with simultaneous eigenvectors)
of $\rho$ and $\sigma$. 

Computing  $\Tr(\rho^\alpha\sigma^{1-\alpha})$ for small values of $1-\alpha,$ we find
\begin{eqnarray*}&&\Tr(\rho^\alpha\sigma^{1-\alpha})=\Tr \e^{\alpha\log\rho}\e^{(1-\alpha)\log\sigma}\\
&&=\Tr\e^{\log\rho} \e^{(\alpha-1)\log\rho}\e^{(1-\alpha)\log\sigma}\\
&&=\Tr\rho(1+(\alpha-1)(\log\rho-\log\sigma)+\mathcal{O}((1-\alpha)^2))\\
&&=1+(\alpha-1)\Tr\rho(\log\rho-\log\sigma)+\mathcal{O}((1-\alpha)^2).\end{eqnarray*}
Thus, $\mathcal{D}_\alpha(\rho\|\sigma)=\Tr\rho(\log\rho-\log\sigma)+\mathcal{O}((1-\alpha)),$
and so when $\alpha\rightarrow1,$ one obtains the {\it von Neumann relative entropy} \cite{von neumann}:
$$\mathcal{D}_1(\rho\|\sigma)=\Tr\rho(\log\rho-\log\sigma).$$

A map $g:H_{n}\times H_{n}\rightarrow\R$, is {\it jointly convex},
if, for $A_1,A_2,B_1,B_2\in H_n$, $0\leq\lambda\leq1$, the following holds,
$$g(\lambda A_1+(1-\lambda)A_2,\lambda B_1+(1-\lambda)B_2)\leq
\lambda g(A_1,B_1)+(1-\lambda)g(A_2,B_2),$$
and $g$ is {\it jointly concave} if $-g$ is jointly convex.

The joint convexity of $\alpha$-RRE for $\alpha\in(0,1)$
is one of its most important properties.
This result was obtained
obtained in \cite{hiai} in a more general context, and
next we give a simple proof. For this purpose,
Lieb's joint concavity Theorem \cite{ruskai} stated in the following Lemma, is needed.
\begin{lema}\label{lieb}
For all matrices $K\in M_n$, $A,B\in H_{n,+}$ and all
$q,r$ such that $0\leq q\leq 1$, $0\leq r\leq 1$ with $q+r\leq1,$
 the real valued function
$$\Tr K^*A^qKB^r$$
is jointly concave in $A,B.$
\end{lema}
\begin{teor}\label{conc}The map ${\cal D}_\alpha:H_{n,+,1}\times H_{n,+}\rightarrow \R$ is jointly convex for $\alpha\in(0,1)$.
\end{teor}
\begin{proof}
Consider in Lemma \ref{lieb},  $r=1-\alpha$,  $q=\alpha$, $\alpha\in(0,1)$ and $K=I_n$.
For $\rho_1,\rho_2\in H_{n,+,1}$, $\sigma_1,\sigma_2\in H_{n,+},$
$0\leq\lambda\leq1$, and the real valued function
$$g(\rho,\sigma)=\Tr\rho^\alpha\sigma^{1-\alpha},$$
the lemma ensures that
$$g(\lambda\rho_1+(1-\lambda)\rho_2,\lambda\sigma_1+(1-\lambda)\sigma_2)\leq
\lambda g(\rho_1,\sigma_1)+(1-\lambda)g(\rho_2,\sigma_2).$$
Since $\log x/(\alpha-1)$ for $\alpha\in(0,1)$ is a decreasing and convex function of $x$, we get
\begin{eqnarray*}&&{\log(g(\lambda\rho_1+(1-\lambda)\rho_2,\lambda\sigma_1+(1-\lambda)\sigma_2))\over\alpha-1}\leq
{\log(\lambda g(\rho_1,\sigma_1)+(1-\lambda)g(\rho_2,\sigma_2))\over\alpha-1}\\
&&\leq
{\lambda\log g(\rho_1,\sigma_1)\over\alpha-1}+{(1-\lambda)\log g(\rho_2,\sigma_2)\over\alpha-1}
,\end{eqnarray*}
and the result follows.
\end{proof}
\begin{corol}The von Neumann map $D_1(\rho\|\sigma):H_{n,+,1}\times H_{n,+,1}\rightarrow\R$ is jointly convex.
\end{corol}
\begin{proof}The result follows taking the limit $\alpha\rightarrow1$ in Theorem \ref{conc}
and recalling that convexity is  preserved in the limit.
\end{proof}
\subsection{A lower bound for $\alpha$-RRE}
The following result extends the well known non negativity
property of von Neumann relative entropy: $D_1(\rho\|\sigma)\geq0$ for $\rho,\sigma$ such that
$\Tr\rho=\Tr\sigma=1.$
\begin{teor}\label{T2.1}
Let $\sigma\in H_{n,+}.$ Then, for
$\rho$ ranging over  $H_{n,+,1},$
$$\mathcal{D}_\alpha(\rho\|\sigma)
\geq-\log\Tr\sigma,~~\alpha\in(0,1)\cup(1,\infty).$$
Equality occurs if and only if $\rho=\sigma/\Tr\sigma.$
\end{teor}
\begin{proof}
For $\alpha<1$, minimizing $\mathcal{D}_\alpha(\rho\|\sigma)$
for a fixed $\sigma$ is equivalent to minimizing
$$\mathcal{T}=\Tr(\rho^\alpha\sigma^{1-\alpha}).$$
For $\alpha>1$, minimizing $\mathcal{D}_\alpha(\rho\|\sigma)$
for a fixed $\sigma$ is equivalent to maximizing $\cal T$.
Next, we optimize $\cal T$.
Suppose that the matrices $\rho,\sigma$
are such that $\cal T$ is optimal. Since the trace is unitarily invariant,
without loss of generality, we can take $\sigma$ in diagonal form.
Then, for $\epsilon>0$ sufficiently small and $S$ arbitrary in $H_n$, we have
$\e^{i\epsilon S}=I_n+i\epsilon S+{\cal O}(\epsilon^2),$ and so
$$\left.{\d\over\d\epsilon}\Tr(\rho^\alpha\e^{i\epsilon S}\sigma^{1-\alpha}\e^{-i\epsilon S})\right|_{\epsilon=0}=
i\Tr S[\sigma^{1-\alpha},\rho^\alpha]= 0,$$
where
$$[\sigma^{1-\alpha},\rho^\alpha]=\sigma^{1-\alpha}\rho^\alpha-\rho^\alpha\sigma^{1-\alpha}.$$
implying that
$$[\rho^\alpha,\sigma^{1-\alpha}]=[\rho,\sigma]=0.$$
As a consequence, the Hermitian matrices $\rho,\sigma$ are simultaneously
unitarily diagonalizable. Since the trace is unitarily invariant, without
loss of generality we may assume  $\rho,\sigma$ in diagonal form,
$\rho={\rm diag}(\rho_1,\ldots,\rho_n)$,
$\sigma={\rm diag}(\sigma_1,\ldots,\sigma_n)$. As  $\log x$ is an increasing
function of the argument $x$, we find
$$\mathcal{T}\leq\sum_{i=1}^n\e^{\alpha\log\rho_i+(1-\alpha)\log\sigma_{i}}
=\sum_{i=1}^n\rho_i^{\alpha}\sigma_{i}^{(1-\alpha)},\quad\alpha\leq1.$$
and
$$\mathcal{T}\geq\sum_{i=1}^n\e^{\alpha\log\rho_i+(1-\alpha)\log\sigma_{i}}
=\sum_{i=1}^n\rho_i^{\alpha}\sigma_{i}^{(1-\alpha)},\quad\alpha\geq1.$$
Thus,
$$\mathcal{D}_\alpha(\rho\|\sigma)\geq{\log\sum_{i=1}^n\rho_i^{\alpha}\sigma_{i}^{(1-\alpha)}\over\alpha-1}
\geq-\log\sum_{j=1}^n\sigma_j.$$
Next, we optimize ${\sum_{i=1}^n\rho_i^{\alpha}\sigma_{i}^{(1-\alpha)}}$ under the constraint
$\sum_{i=1}^n\rho_i=1$, using Lagrange multipliers techniques. We consider the function
$$\psi={\sum_{i=1}^n\rho_i^{\alpha}\sigma_{i}^{(1-\alpha)}}-\lambda\left(\sum_{i=1}^n\rho_i-1\right),~~\lambda\in \R.$$
The extremum condition leads to
$${\partial\psi\over\partial\rho_i}=\alpha\rho_i^{\alpha-1}\sigma_i^{1-\alpha}-\lambda=0,$$
so that
$$\rho_i=\left({\lambda\over\alpha}\right)^{1/(\alpha-1)}\sigma_i.$$
The Lagrange multiplier $\lambda$ is determined observing that $\sum_{i=1}^n\rho_i=1.$
Thus, $(\lambda/\alpha)^{1/(\alpha-1)}=1/ \sum_{i=1}^n\sigma_i $ and so
$\rho_i=\sigma_i/\sum_{j=1}^n\sigma_j$. The asserted result is finally obtained.
\end{proof}

\section{The R\'enyi-Peierls-Bogoliubov inequality}\label{S3}

In statistical mechanics, the absolute temperature is usually denoted by $T$,
and its inverse, $1/T$, by $\beta$. The {\it internal energy} is defined as the expectation value of the
Hamiltonian $H$ in  the state $\rho$, i.e., $\Tr\rho H$. Here, we are assuming that $\beta=1/T=1$.

The $\alpha$-{\it expectation value} of the Hermitian operator $H$ is defined and denoted as
$$\langle H\rangle_\alpha:={1\over\alpha-1}\log{\Tr\rho^\alpha\e^{(\alpha-1)H}\over\Tr\rho^{\alpha}},~~ \alpha\in(0,1)\cup(1,\infty).$$
where $\rho\in H_{n,+,1}.$

We define, for $\beta$=1, the $\alpha$-{\it R\'enyi internal energy}
($\alpha$-RIE) as the $\alpha$-expectation value of $H$ in the state $\rho$,
\begin{equation}\label{Ealpha}E_\alpha(\rho,H):=\langle H\rangle_\alpha=
{1\over\alpha-1}\log{\Tr\rho^\alpha\e^{(\alpha-1)H}\over\Tr\rho^\alpha},~~\alpha\in(0,1)\cup(1,\infty).\end{equation}

We remark that some authors define differently the $\alpha$-RIE, according to
$$
{\Tr\rho^\alpha{H}\over\Tr\rho^\alpha},$$
that is, as the average of $H$ in the state $\rho^\alpha$. In the definition we
are proposing, the logarithm of the average
of $\e^{(\alpha-1)H}$ in state $\rho^\alpha$, multiplied by $1/(\alpha-1),$ is considered.
This option considerably simplifies the formalism involved in the thermodynamical considerations.
On the other hand, it may be easily shown that, for $\alpha\rightarrow1$, $E_\alpha(\rho,H)$
approaches the standard expectation value of the Hamiltonian
and of the internal energy arising in statistical thermodynamics,
$$E_1(\rho,H)=\lim_{\alpha\rightarrow1}E_\alpha(\rho,H)=\Tr\rho H.$$

We define (for $\beta=1$) the $\alpha$-{\it R\'enyi free energy} ($\alpha$-RFE) as
$$F_\alpha(\rho,H):=E_\alpha(\rho,H)-S_\alpha(\rho)={\log\Tr\rho^\alpha\e^{(\alpha-1)H}\over\alpha-1}
,~~\alpha\in(0,1)\cup(1,\infty).$$
Notice that $F_\alpha(\rho,H)$ is closely related to the $\alpha$-RRE, as
$$F_\alpha(\rho,H)=\mathcal{D}_\alpha(\rho\|\e^{-H}).$$

According to the principles of thermodynamics, the state of equilibrium of a system
is the one for which the free energy is minimized, at constant temperature.
The {\it Helmholtz state}, which is the equilibrium state, is obtained by minimizing
the {Helmholtz free energy} (for fixed temperature).

The next Theorem characterizes, from the knowledge of $H$,
the state which minimizes the $\alpha$-RFE, the so called
the {\it equilibrium state} of the system.
This result is
also known as the Helmholtz free energy variational principle.
\begin{teor}\label{T3.1}(R\'enyi-Peierls-Bogoliubov inequality)
Let $H\in H_n$ be given and  $\rho\in H_{n,+,1}$ be arbitrary. Then,
$$F_\alpha(\rho,H)\geq-\log\Tr\e^{-H},~~\alpha\in(0,1)\cup(1,\infty).$$
Equality occurs if and only if $\rho=\e^{-H}/\Tr\e^{-H}.$
\end{teor}
\begin{proof}
Replacing in Theorem \ref{T2.1} $\sigma$ by $\e^{-H}$, the result follows.
\end{proof}

If the state of equilibrium $\rho$ is known, then the Hamiltonian
of the system is obtained as
$H=-\log\rho-\log\Tr\e^{-H}I_n$, where $I_n\in M_n$ is the identity matrix.

Consider $H$ as a perturbation of the Hamiltonian $H_0$. So, $H_0$ may be regarded
as a convenient approximation of $H.$ The following result provides
useful information on $\Tr\e^{-H}$ from $\Tr\e^{-H_0}$.
\begin{corol}
For $H,H_0\in H_n,$ we have
$${1\over\alpha-1}\log{\Tr\e^{-\alpha H_0}\e^{(\alpha-1)H}\over\Tr\e^{-H_0}}\geq-\log{\Tr\e^{-H}\over\Tr\e^{-H_0}}.$$
\end{corol}
\begin{proof}Considering, in Theorem \ref{T3.1}, $\rho=\e^{-H_0}/\Tr\e^{-H_0},$
the result follows by a trivial computation.
\end{proof}
\section{Partition function  and the R\'enyi entropy}\label{S4}
The {\it partition function}
\begin{equation}\label{Z}Z_\beta=\Tr\e^{-\beta H},\end{equation}
where $\beta=1/T$ denotes the inverse of the {\it absolute temperature} and
$H$ is the Hamiltonian of the physical system, plays a fundamental role
in standard statistical thermodynamics.
The discussion of some issues requires the consideration of
the parameter $\beta$, so we will relax the restriction $\beta=1$, which has been adopted up to now.
In standard statistical thermodynamics, the equilibrium properties of the system
are encapsulated into the logarithm of the partition function.
In particular,  the {\it internal energy} $$E_\beta={\Tr H\e^{-\beta H}\over \Tr\e^{-\beta H}}$$
is related to the derivative of
$\log Z_\beta$ with respect to $\beta$ as
$$E_\beta=-{\d \log Z_\beta\over\d\beta}.$$
So, the following question naturally arises. What is the relation between the internal
energy and the partition function in the context of R\'enyi thermodynamics?
Notice that in R\'enyi thermodynamics the partition function is as meaningful
as in standard statistical mechanics, because the expression of the equilibrium state
in the R\'enyi thermodynamics coincides with the corresponding expression in the von Neumann setting,
$\rho=\rho_0:={\e^{-H}/\Tr\e^{-H}}$.

Next we derive a relation between the internal energy and $\log Z_\beta,$ in R\'enyi's thermodynamics.
For this purpose we define the $\alpha$-{\it derivative} of the function $\psi:\R\rightarrow\R$ as the quotient
$${\psi(\beta\alpha)-\psi(\beta)\over \beta(\alpha-1)}.$$
Replacing $\rho$ by $\e^{-H}/\Tr\e^{-H}$  in (\ref{Ealpha}), we conclude, from Corollary (\ref{T3.1}),
that the R\'enyi {\it equilibrium internal energy} (for $\beta=1$) reduces to
\begin{equation}\label{*****}E_\alpha(\rho_0,H)
={1\over\alpha-1}(\log\Tr\e^{-H}-\log\Tr\e^{-\alpha H}).\end{equation}
\begin{pro}\label{P4.1}
For $\beta=1$, the R\'enyi equilibrium internal energy 
is the $\alpha$ derivative of $-\log Z_\beta$, taken at $\beta=1$.
\end{pro}
\begin{proof}Having in mind (\ref{*****}) and that the logarithm of the
partition function, for arbitrary $\beta$, reads
$$\log Z_\beta=\log\Tr\e^{-\beta H},$$
we get
$$E_\alpha(\rho_0,H)
=-{\log Z_\alpha-\log Z_1\over\alpha-1}
=\left.-{\log Z_{\beta\alpha}-\log Z_\beta\over \beta(\alpha-1)}\right|_{\beta=1},$$
and the result follows.
\end{proof}

Since $-\log Z_1=E_\alpha(\rho_0,H)-S_\alpha(\rho_0),$ the relation between
the partition function and the internal energy also determines the entropy $S_\alpha(\rho_0)$.

The discussion in this Section is analogous to the arguments in \cite{baez}.
\section{R\'enyi maximum entropy principle}\label{S5}
In order to formulate the {\it maximum entropy principle} (MaxEnt) in the context of
R\'enyi thermodynamics we introduce the concept of R\'enyi
{\it internal energy} for $\beta\neq1$, as a generalization of (\ref{Ealpha})
\begin{equation}\label{Ealphabeta}E_{\alpha,\beta}(\rho,H):={1\over\beta}\langle\beta H\rangle_\alpha={\log\Tr\rho^\alpha\e^{(\alpha-1)\beta H}-\log\Tr\rho^\alpha\over\beta(\alpha-1)}
.\end{equation}
The parameter $\beta$ controls, or tunes, the internal energy.

\begin{pro}\label{P5.1}
For arbitrary $\beta$, the R\'enyi equilibrium internal energy
is the $\alpha$ derivative of $-\log Z_\beta$.
\end{pro}
\begin{proof}
For $\beta\neq1$, the R\'enyi equilibrium  internal energy  reduces to
$$E_{\alpha,\beta}(\rho_0,H)
=
{\log\Tr\e^{-\beta H}-\log\Tr\e^{-\alpha\beta H}\over\beta(\alpha-1)}={\log Z_{\alpha\beta}-\log Z_\beta\over\beta(\alpha-1)},$$
and the result follows.\end{proof}
\begin{pro}The R\'enyi equilibrium internal energy is a monotonously decreasing function of $\beta$.
\end{pro}
\begin{proof}
We have that $-\log Z_\beta$ is a convex function of $\beta$
 as $\log Z_\beta$ is concave, because
$${\d^2 \log Z_\beta\over\d \beta^2}={\Tr H^2\e^{-\beta H}\over\Tr\e^{-\beta H}}
-\left({\Tr H\e^{-\beta H}\over\Tr\e^{-\beta H}}\right)^2\geq0.$$
Now, observing that equality occurs only in the limit $\beta\rightarrow\infty$, we conclude that
$E_{\alpha,\beta}(\rho_0,H)$, being, according to Proposition \ref{P5.1}, the slope of the secant line through the points
$(\beta,-\log Z_\beta)$ and $(\alpha\beta,-\log Z_{\alpha\beta}),$
decreases as $\beta$ increases.\end{proof}

We remark that $E_{\alpha,\beta}(\rho_0,H)$ lies in the interval defined by
the lowest and the highest eigenvalue of $H$. This follows, observing that for $\beta=\pm\infty$
these eigenvalues are reached,
$$\lambda_{min}(H)\leq E_{\alpha,\beta}(\rho_0,H)\leq\lambda_{max}(H).$$
For $\rho\neq\rho_0,$ $E_{\alpha,\beta}(\rho,H)$ may not be in that interval.
\begin{pro}For $\rho_0$ the $\beta$-dependent equilibrium state, and for $\lambda_{min}(H)$, $\lambda_{max}(H)$ the
lowest and the highest eigenvalue of $H$, respectively, we have
$$\lim_{\beta\rightarrow\infty}E_{\alpha,\beta}(\rho_0,H)=\lambda_{min}(H),\quad\lim_{\beta\rightarrow-\infty}E_{\alpha,\beta}(\rho_0,H)=\lambda_{max}(H).$$
\end{pro}
\begin{proof}The result follows keeping in mind the convexity of $-\log Z_\beta$ and that
$\lambda_{min}(H)$,
$\lambda_{max}(H)$
are the slopes of the asymptotes to $-\log Z_\beta$.
\end{proof}

For $\beta\neq1$, we define the $\alpha$-{\it R\'enyi free energy} $F_{\alpha,\beta}(\rho,H)$ as
$$F_{\alpha,\beta}(\rho,H):=E_{\alpha,\beta}(\rho,H)-{1\over\beta}S_\alpha(\rho)
={\log\Tr\rho^\alpha\e^{(\alpha-1)\beta H}\over\beta(\alpha-1)}.$$
The maximum entropy principle states that the equilibrium state $\rho$ is
obtained by maximizing $-\beta F_{\alpha,\beta}(\rho,H)$ with respect to $\rho$,
under the constraint $\sum_{i=1}^n \rho_i=1$, which, for $\beta\geq0,$
is equivalent to minimizing $F_{\alpha,\beta}(\rho,H)$ under the same constraint.
Replacing in Theorem \ref{T3.1} $\sigma$ by $\e^{-\beta H}$, we obtain
\begin{teor} (R\'enyi-Peierls-Bogoliubov-inequality)
Let $H\in H_n$ be given and $\rho\in H_{n,+,1}$ be arbitrary. Then,
$$\beta F_{\alpha,\beta}(\rho,H)\geq-\log\Tr\e^{-\beta H},~~\alpha\in(0,1)\cup(1,\infty),~\beta\in\R.$$
Equality occurs if and only if $\rho=\e^{-\beta H}/\Tr\e^{-\beta H}.$
\end{teor}
This result is in agreement with the corresponding expression in von Neumann statistical mechanics.
We observe that, in conventional thermodynamics, $\beta\geq0.$ However, if $n$ is finite, it is also meaningful to consider $\beta<0.$

The equilibrium
state depends only on the value of the parameter $\beta$, which is determined by
the required value of the internal energy.
\section{Uncertainty relations}\label{S6}
The {\it uncertainty principle} was formulated by Heisenberg in 1927 and states that
it is not possible to measure simultaneously, with absolute precision, the {\it position operator} $x$
and the {\it momentum operator} $p$ of a particle. These operators are considered in
the one dimensional context. The product of the uncertainties
in the respective measurements $\Delta x$ and $\Delta p$,
is of the order of Plank's constant $\hbar$. We consider units such that $\hbar=1.$
This indeterminacy relation may be formulated in precise mathematical form as
$$\Delta x\Delta p\geq {\hbar\over2}.$$

The Heisenberg-Robertson uncertainty principle, firstly proposed by Heisenberg
and then generalized by Robertson \cite{robertson} gives a lower bound
for the product of the standard deviation of two observables.
To state it, we introduce some useful concepts.
For $A\in H_n,$ the {\it expectation value} of the measurement of
the observable $A$ in the state $\rho\in H_{n,+,1}$  is
$$\langle A\rangle={\Tr\rho A\over\Tr\rho}.$$
The {\it variance} in the measurement of $A$ is defined as
$$\sigma_A^2={1\over\Tr\rho}{\Tr\rho}(A-\langle A\rangle)^2.$$ 
The {\it uncertainty} in the measurement of $A$ is defined as
the {\it standard deviation} $\sigma_A$. As usually,
we denote the anticommutator of $A,B$ as
$$\{A,B\}=AB+BA .$$
The {\it covariance} of $A,B\in H_n$ is determined as
$${\rm Cov}({A,B})={1\over\Tr\rho}{\Tr\rho}\left({1\over2}\{ A,B\}-\langle A\rangle\langle B\rangle\right).$$
Observe that ${\rm Cov}({A,A})=\sigma_A^2,$ i.e., the variance is a particular case of the covariance, and ${\rm Cov}(A,B)={\rm Cov}(B,A).$
The Heisenberg-Robertson uncertainty relation states that
$$\sigma_A^2\sigma_B^2\geq{1\over4}|\langle[A,B]\rangle|^2,$$
and was improved by Schr\"odinger as
$$\sigma_A^2\sigma_B^2\geq{1\over4}|\langle[A,B]\rangle|^2+{1\over4}\langle\{A,B\}-\langle A\rangle\langle B\rangle\rangle^2.$$

The following theorem gives a lower bound for the product of
the standard deviations of two quantum obsevables:
\begin{teor}\label{T4.1} Let $A$ and $B$ be Hermitian matrices
and $\D\in H_{n,+,1}$. Then,
\begin{equation}\label{****}\sigma_A^2\sigma_B^2\geq
{\rm Cov}({A,B})^2+\left({1\over2}\langle i[A,B]\rangle\right)^2.\end{equation}
Equality occurs if and only if $A$ is a multiple of $B$.
\end{teor}
\begin{proof} We observe that $i[A,B]$ is Hermitian, as $(i[A,B])^*=i[A,B]$.
Let
$${A}'_\rho:={\rho^{1/2}\over(\Tr\rho)^{1/2}}(A-\langle A\rangle),
~~{B}'_\rho:={\rho^{1/2}\over(\Tr\rho)^{1/2}}(B-\langle B\rangle).$$

We easily find
$$\sigma_A^2=\Tr({A'}_\rho{A'}_\rho^*),~~\sigma_B^2=\Tr({B'}_\rho{B'}_\rho^*)$$
and
$$\Tr({A'}_\rho{B'}_\rho^*)={1\over\Tr\rho}\Tr\rho(A-\langle A\rangle)(B-\langle B\rangle)
={1\over\Tr\rho}\Tr\rho\left({1\over2}\{B,A\}+{1\over2}[B,A]-\langle A\rangle\langle B\rangle\right)
.$$
On the other hand,
$$\Tr({B'}_\rho{A'}_\rho^*)={1\over\Tr\rho}\Tr\rho\left({1\over2}\{B,A\}-{1\over2}[B,A]-\langle A\rangle\langle B\rangle\right),$$
so that
$$\Tr({A'}_\rho{B'}_\rho^*)=
{1\over2}\left\langle\{ B,A\}\right\rangle-\langle A\rangle\langle B\rangle+{1\over2i}\langle i[B,A]\rangle,$$
and
$$\Tr({B'}_\rho{A'}_\rho^*)=
{1\over2}\left\langle \{B,A\}\right\rangle-\langle A\rangle\langle B\rangle-{1\over2i}\langle i[B,A]\rangle.
$$
According to the matricial Schwartz inequality, we have
$$
\Tr({A}'_\rho{A'}^*_\rho)\Tr({B}'_\rho{B'}_\rho^*)
\geq\Tr({A'}_\rho{B'}_\rho^*)\Tr({B'}_\rho{A'}_\rho^*).$$
Equality occurs if and only if $A'_\rho$ is a multiple of $B'_\rho$,
that is,  if and only if $A$ is a multiple of $B$.
\end{proof}

We present the relation (\ref{****}) in a form susceptible of extension.
Let us introduce the {\it covariance matrix}
$$\sigma({A,B})=\left[\begin{matrix}\sigma_A^2&{\rm Cov}({A,B})\\
{\rm Cov}({A,B})&\sigma_B^2\end{matrix}\right].$$
The inequality in  (\ref{****}) can be expressed as
$$\det\sigma({A,B})\geq\left({1\over2}\langle i[A,B]\rangle\right)^2.$$

For $m$ observables $\{X_k\}_{k=1}^m$,
let $${X'_j}_\rho:=({{\Tr\rho}})^{-1/2}{\rho^{1/2}}(X_j-\langle X_j\rangle),\quad j=1,\ldots,m.$$
Then
$$\Tr {X'_j}_\rho{X'_k}_\rho^*={1\over\Tr\rho}{\Tr\rho}\left( X_jX_k-\langle X_j\rangle\langle X_k\rangle\right)
={\rm Cov}(X_j,X_k)-{i\over2}\langle i[X_j,X_k]\rangle$$
where
$$\langle i[X_j,X_k]\rangle={1\over\Tr\rho}{\Tr\rho}\left( i[X_j,X_k]\right).$$

Notice that
\begin{equation}\label{Cov}{\rm Cov}(X_j,X_k)={1\over2}(\Tr {X'_j}_\rho{X'_k}_\rho^*+\Tr {X'_k}_\rho{X'_j}_\rho^*)\end{equation}
and
$$-{i\over2}\langle i[X_j,X_k]\rangle={1\over2}(\Tr {X'_j}_\rho{X'_k}_\rho^*-\Tr {X'_k}_\rho{X'_j}_\rho^*).$$
We consider the $m\times m$ covariance matrix
\begin{equation}\label{sigma}\sigma({X_1,\ldots,X_m})=\left[\begin{matrix}{\rm Cov}(X_1,X_1)&\dots&{\rm Cov}(X_1,X_m)\\
\vdots&\ddots&\vdots\\
{\rm Cov}(X_m,X_1)&\dots&{\rm Cov}(X_m,X_m)\end{matrix}\right]\end{equation}
and the matrix
formed by the measurements of the commutators of the observables,
\begin{equation}\label{delta}\delta({X_1,\ldots,X_m})=\left[\begin{matrix}-{i\over2}\langle i[X_1,X_1]\rangle&\dots&-{i\over2}\langle i[X_1,X_m]\rangle\\
\vdots&\ddots&\vdots\\
-{i\over2}\langle i[X_m,X_1]\rangle&\dots&-{i\over2}\langle i[X_m,X_m]\rangle
\end{matrix}\right].\end{equation}

The $m\times m$ matrix
$$\tau=
\left[\begin{matrix}\Tr {X'_1}_\rho{X'_1}_\rho^*&\dots&\Tr{X'_1}{X'_m}_\rho^*\\
\vdots&\ddots&\vdots\\
\Tr{X'_m}_\rho{X'_1}_\rho^*&\dots&\Tr{X'_m}_\rho{X'_m}_\rho^*
\end{matrix}\right].$$
is positive semidefinite, as $z^*\tau z\geq0$ for any $z\in\C^m$. 
In fact, $\tau$ may be seen as the Gram matrix of the operators $X_{kl}$ with respect to
the Hilbert-Schmidt inner product $\langle Y,X\rangle=\Tr X^* Y.$
Obviously, $\tau=\sigma+\delta.$
\begin{teor}\label{T4.2} For $\sigma(X_1,\ldots,X_m),~\delta(X_1,\ldots,X_m)$ in (\ref{sigma}), (\ref{delta}),
such that $\sigma(X_1,\ldots,X_m)+\delta(X_1,\ldots,X_m)$ is positive definite
and $m$ an even number, we have
\begin{equation}\label{m op}\det\sigma(X_1,\ldots,X_m)>\det i\delta(X_1,\ldots,X_m).\end{equation}
\end{teor}
To prove this result we present an auxiliary Lemma.
\begin{lema}\label{lemma}For $C$ a positive definite matrix with even dimension,
with $A=(C+C^T)/2$ and $B=(C-C^T)/2$,  we have
$$\det A>\det iB.$$
\end{lema}
\begin{proof}
By hypothesis, $C$ is positive definite, so
it is clear that $B\in H_n$ and $A$ is positive definite.
We consider the characteristic polynomial
$$\det(\lambda A+B).$$
Since $A$ is symmetric and $B$ antisymmetric, the condition $\det(\lambda A+B)=0$
implies $\det(\lambda A-B)=\det(\lambda^2 A^2-B^2)=0$, so that the characteristic roots occur in symmetric pairs.
Let $U$ be a unitary matrix such that
$$U^*{A^{-1/2}}B{A^{-1/2}}U={\rm diag}({\lambda_1,\ldots\lambda_m}).$$
Then
$$B= A^{1/2} U {\rm diag}({\lambda_1,\ldots\lambda_m})U^* A^{1/2} ,\quad A= A^{1/2} UU^*  A^{1/2}, $$
and
$$C=A+B=A^{1/2}U {\rm diag}({1+\lambda_1,\ldots1+\lambda_m})U^* A^{1/2}$$
implying that $\lambda_1,\ldots,\lambda_m\in[-1,1].$
Thus,
$${\det}(U^*{ A^{-1/2}}iB{ A^{-1/2}}U) ={\det iB\over\det A}=(-1)^{m/2}\lambda_1\ldots\lambda_m <1.$$
Observing that $(-1)^{m/2}\lambda_1\ldots\lambda_m >0$, the result follows.\end{proof}

The proof of Theorem \ref{T4.2} is a simple consequence of Lemma \ref{lemma},
observing that $\sigma({X_1,\ldots,X_m})=(\tau+\tau^T)/2$ and $\delta({X_1,\ldots,X_m})=(\tau-\tau^T)/2$.
\begin{corol} For $\delta(X_1,\ldots,X_m)$ in (\ref{delta}) and $\sigma_j^2={\rm Cov}(X_j,X_j)$
in (\ref{Cov}),
$$\prod_{j=1}^m\sigma_j^2\geq\det i\delta(X_1,\ldots,X_m).$$
\end{corol}
\begin{proof}Notice that $\sigma(X_1,\ldots,X_m)$ in (\ref{sigma}) is positive definite.
From Theorem (\ref{T4.1}) and Hadamard determinantal
inequality, we obtain
$$\prod_{j=1}^m\sigma_j^2\geq\det \sigma(X_1,\ldots,X_m),$$
and the result follows.
\end{proof}

\subsection{$\alpha$-variance}
The $\alpha$-{\it expectation value} of the Hermitian operator $A$ has been defined as
$$\langle A\rangle_\alpha={1\over\alpha-1}\log{\Tr\rho^\alpha\e^{(\alpha-1)A}\over\Tr\rho^{\alpha}},$$
where $\rho\in H_{n,+,1}.$
If $A>0$, then $\langle A\rangle_\alpha>0.$
The $\alpha$ expectation value is strongly non linear. We observe that, for $A,B\in H_n,~~\lambda\in\R,$
we have, in general,
$$\langle \gamma A\rangle_\alpha\neq\gamma\langle A\rangle_\alpha$$
and
$$\langle A+B\rangle_\alpha\neq\langle A\rangle_\alpha+\langle B\rangle_\alpha$$
except for
$$\langle\gamma I_n\rangle_\alpha=\gamma=\gamma\langle I_n\rangle_\alpha$$
and
$$\langle A+\gamma I_n\rangle_\alpha=\langle A\rangle_\alpha+\gamma.$$
Notice that $\langle(A-\langle A\rangle_\alpha I_n)\rangle_\alpha=0$ and that
$(A-\langle A\rangle_\alpha I_n)^2>0.$
The $\alpha$-{\it variance} in the measurement of $A$ may be naturally defined as
$$\sigma_{A,\alpha}^2=\Tr\rho_0(A-\langle A\rangle_\alpha I_n)^2.
$$
This definition is consistent with the one for $\alpha=1$, since $\lim_{\alpha\rightarrow1}\sigma_{A,\alpha}^2=\sigma_A^2.$
The derivation and physical interpretation of inequalities
analogous to (\ref{****}) and (\ref{m op}), for $\alpha\in(0,1)$, remains an open problem.
\begin{Exa}
We consider the Hamiltonian
$$H={\rm diag} ({3, 2, 4, 1, 5, 9, 2, 6, 5, 3, 5, 9})\in M_{12},$$
and compute, in the equilibrium state, for $\beta=1$, the $\alpha$-expectation value
and the $\alpha$-standard deviation
for $\alpha\in\{0,1/2,1,2,\infty\}.$ We have found the following values

\noindent
$\alpha=0$,~~
$\langle H\rangle_\alpha=2.73416$,
$\sigma_{H,\alpha}=1.325389$.\\
$\alpha=1/2$, $\langle H\rangle_\alpha=2.11338$,
$\sigma_{H,\alpha}=1.02608$.\\
$\alpha=1$,~~ $\langle H\rangle_\alpha=1.79549$,
$\sigma_{H,\alpha}=1.39549$.\\
$\alpha=2$,~~ $\langle H\rangle_\alpha=1.48008$,
$\sigma_{H,\alpha}=1.02531.$\\
$\alpha=\infty$,~ $\langle H\rangle_\alpha=1$,
$\sigma_{H,\alpha}=1.2588$.\\

\noindent
The equilibrium free energy $F_\alpha(\rho_0,H)=0.249258$ does not depend on $\alpha$,
so, the entropy of the equilibrium state  $S_\alpha(\rho_0)=\langle H\rangle_\alpha+F_\alpha(\rho_0,H)$
has also been determined.   For $\alpha=\infty$, $\langle H\rangle_\alpha$ is equal to the lowest eigenvalue of $H$.
We notice that $\langle H\rangle_\alpha$ decreases as $\alpha$ increases,
and that the $\alpha$-standard deviation of the measurement of $H$ is highest for $\alpha=1$.
\end{Exa}
\section{Discussion}\label{S7}
We have presented self-contained proofs of fundamental inequalities in the setting
of R\'enyi's statistical thermodynamics, which is formulated through the replacements,
of $\langle\beta H\rangle_1$  and of $S_1(\rho)$,
in the expression of the free energy, respectively, by
$\langle\beta H\rangle_\alpha$ and $S_\alpha(\rho)$, for $\alpha$ a parameter in $(0,1)\cup(1,\infty)$.
Definitions for thermodynamical quantities, such as free energy, entropy and partition function were given.
We adopted the paradigm in \cite{misra,skrzypczyk} for dealing with thermodynamical processes in the
framework of quantum theory. 
By assuming
the laws of thermodynamics, the equilibrium state of a given system
is determined. The R\'enyi MaxEnt principle has been stated and the
equilibrium state has been determined.

Uncertainty relations have been revisited in the present context.
It has been shown that the product of the uncertainties on the measurements of an even number of
observables can not be less than a certain function of their commutators.
This extends the uncertainty principles of Heisenberg and its refinement by Schr\"odinger,
who introduced the correlations of two observables. The statement of these principles
in R\'enyi's statistical thermodynamics is an open problem.

Different types of uncertainty relations have been considered.
There are many ways to quantify the uncertainties of  measurements.
The lower bound in the Heisenberg-Robertson formulation can happen to be zero,
and so having a global state independent lower bound may be desirable.

Entropic uncertainty
relations have significant importance within quantum information
providing the foundation for the security of quantum cryptographic protocols.
Using majorization techniques, explicit lower bounds for the sum of
R\'enyi entropies describing probability distributions have been derived.
Some results admit generalizations to arbitrary mixed states.

For $\alpha\in(0,1)\cup(1,\infty)$ and $\rho,\sigma\in H_{n,+}$, the {\it sandwiched}
$\alpha$-RRE is defined as
$${\cal D}_\alpha(\rho\|\sigma)={1\over\alpha-1}\log\left(\Tr\left(\sigma^{(1-\alpha)\over\alpha}
\rho\sigma^{(1-\alpha)\over\alpha}\right)^\alpha\right)$$
and reduces to the $\alpha$-RRE when $\alpha$ and $\rho$ commute. The problems we have discussed
may also be considered in the context of this entropy.

A demanding avenue of research is the study of operator R\'enyi entropic inequalities.

\section*{Acknowledgments} This work was partially
supported 
by the European Regional Development Fund through the
program COMPETE and by the Portuguese
Government through the FCT - 
{Funda\c c\~ao} para a Ci\^encia e a Tecnologia under the projects
PEst-C/MAT/UI0324/2011
and UID/FIS/04564/2016.

\end{document}